\newif\ifshort  \newif\iflong
     \shorttrue \longfalse
\documentclass[orivec]{llncs}
\pdfoutput=1 % For forcing PDFLaTeX

\usepackage[usenames,dvipsnames]{color}
\usepackage{amssymb}
\usepackage{graphicx}
\usepackage{subfigure}
\usepackage{wrapfig}
\usepackage{xspace}
\iflong
\title{On the Edge-length Ratio of Outerplanar graphs \thanks{Research supported in part by some very good wine, lots of pizza, and some pasta (especially carbonara).}}
\else
\title{On the Edge-length Ratio of Outerplanar Graphs}
\fi
\author{Sylvain Lazard\inst{1} \and
William Lenhart\inst{2} \and
Giuseppe Liotta\inst{3}
}

\institute{Inria, CNRS, U. Lorraine, France,~\email{sylvain.lazard@inria.fr}\and
Williams College, US,~\email{wlenhart@williams.edu}\and
U. of Perugia, Italy,~\email{giuseppe.liotta@unipg.it}}

\begin{document}

\maketitle

\begin{abstract}

We show that any outerplanar graph admits a planar straight-line drawing such that  the length ratio
of the longest to the shortest edges is strictly less than $2$.
This result is tight in the sense that for any $\epsilon > 0$ there are outerplanar graphs
that cannot be drawn with an edge-length ratio smaller than $2 - \epsilon$.
We also show that every bipartite outerplanar graph
has a planar straight-line drawing with edge-length ratio $1$, and that, for any $k \geq 1$, there
exists an outerplanar graph with a given combinatorial embedding such that  any
planar straight-line drawing has edge-length ratio greater than~$k$.

\end{abstract}

\section{Introduction}\label{se:intro}

The problem of computing a planar straight-line drawing with prescribed edge lengths has been addressed by several authors, partly for its theoretical interest and partly  for its application in different areas, including VLSI, wireless sensor networks, and computational geometry (see for example~\cite{DBLP:journals/siamdm/BattistaV96,DBLP:conf/infocom/DohertyPG01,DBLP:series/natosec/HeldKRV11,DBLP:conf/usenix/SavareseRL02}). Deciding
whether a graph admits a straight-line planar drawing with prescribed edge lengths was shown to be
NP-hard by Eades and Wormald for 3-connected planar graphs ~\cite{DBLP:journals/dam/EadesW90}.
In the same paper, the authors show that it is NP-hard to determine whether a 2-connected planar graph has a {\em unit-length} planar straight-line drawing; that is, a drawing in which all edges have the same length.  Cabello et al. extend this last result by showing that it is NP-hard to decide whether a 3-connected planar graph admits a unit-length planar straight-line drawing \cite{DBLP:journals/jgaa/CabelloDR07}. In addition, Bhatt and Cosmadakis prove that deciding whether a degree-4 tree has a planar drawing such that all edges have the same length and the vertices are at integer grid points is also NP-hard~\cite{DBLP:journals/ipl/BhattC87}.

These hardness results have motivated the study of relaxations and variants of the problem of computing straight-line planar drawings with constraints on the edge lengths.
For example, Aichholzer et al.~\cite{DBLP:conf/cccg/AichholzerHKR14} study the problem of computing straight-line planar drawings where, for each pair of edges of the input graph $G$, it is specified which edge must be longer. They characterize families of graphs that are
{\em length universal}, i.e. they admit a planar straight-line drawing for any given total order of their edge lengths.

Perhaps one of the most natural variants of the problem in the context of graph drawing is that where, instead of imposing constraints on the edge lengths, one aims at computing planar straight-line drawings where the variance of the lengths of the edges is minimized. See for example ~\cite{DBLP:books/ph/BattistaETT99}, where this optimization goal is listed among the most relevant aesthetics that impact the readability of a drawing of a graph.
Computing straight-line drawings where the ratio of the longest to the shortest edge is close to $1$ also arises in the approximation of unit disk graph representations, a problem of interest in the area of wireless communication networks
(see, e.g. ~\cite{Chen2011,DBLP:conf/dialm/KuhnMW04}).

Discouragingly, Eades and Wormald observe in their seminal paper that the NP-hardness of computing 2-connected planar straight-line drawings with unit edge lengths persists even when a small tolerance (independent of the problem size) in the length of the edges is allowed. To our knowledge, little progress has been made on bounding the ratio between the longest and shortest edge lengths in planar straight-line drawings.  We recall the work of Hoffmann et al.~\cite{DBLP:conf/cccg/HoffmannKKR14}, who compare different drawing styles according to different quality measures including the edge-length variance.

In this paper we study planar straight-line drawings of outerplanar graphs that bound
the ratio of the longest to the shortest edge lengths from above by a constant.
% the variance in the lengths of its edges is bounded by a constant.
We define the {\em planar edge-length ratio} of a planar graph $G$ as the smallest ratio between the longest and the shortest edge lengths over all planar straight-line drawings of $G$. The main result of the paper is the following.

\begin{theorem}\label{th:main}
The planar edge-length ratio of an outerplanar graph is strictly less than $2$. Also, for any given real positive number $\epsilon$, there exists an outerplanar graph whose planar edge-length ratio is greater than $2 - \epsilon$.
\end{theorem}

Informally, Theorem~\ref{th:main} establishes that $2$ is a tight bound for the planar edge-length ratio of outerplanar graphs.
The upper bound is proved by using a suitable decomposition of an outerplanar graph into subgraphs
called {\em strips}, then drawing the graph strip by strip.
The lower bound is proved by taking into account all possible planar embeddings of a maximal outerplanar graph whose maximum vertex degree is a function of $\epsilon$.  Theorem~\ref{th:main} naturally suggests some interesting questions that are discussed in
Section~\ref{se:openproblems}.

We shall assume familiarity with basic definitions of graph planarity and of graph drawing~\cite{DBLP:books/ph/BattistaETT99} and introduce only the terminology and notation that is strictly needed for our proofs.

Note: Some proofs have been moved to the Appendix.

\section{Proof of Theorem~1}\label{se:proof-Th.1}

It suffices to establish the result for maximal outerplanar graphs. To show that the edge-length ratio of a maximal outerplanar graph is always less than $2$,
we imagine decomposing the dual $G^*$ of $G$ into a set of disjoint paths, which we call {\em chains}.
Each chain corresponds to some sequence of pairwise-adjacent triangles of $G$.
The set of chains inherits a tree structure from $G^*$, and we use this structure to direct an  algorithm that draws each of the chains proceeding from the root of this tree down to its leaves.
We formally define a chain to be a sequence $T_s, T_{s+1}, \ldots, T_t$ of triangles of $G$ where $s \leq 0 \leq t$, and such that
(i) $T_0$ consists of an outer edge of $G$ whose vertices are labeled with $0$, along with a third vertex labeled $1$, 
(ii) for each $i : 1 \leq i \leq t$, the vertices of $T_i$ are labeled by $\{i-1, i, i+1\}$ so that
$T_i$ and $T_{i-1}$ share the edge having vertices labeled $i$ and $i-1$, and
(iii) for each $i : -1 \geq i \geq s$, the vertices of $T_i$ are labeled by $\{ i, i+1, i+2\}$ so that
$T_i$ and $T_{i+1}$ share the edge having vertices labeled $i+1$ and $i+2$.
Note that this definition prohibits {\em fans} (consecutive triangles all sharing a common vertex)
containing more than $3$ triangles, except for the vertex labeled $1$, which has four
incident triangles on the chain.

The decomposition into chains is constructed by first selecting an edge $e'$ on the outer face of some outerplanar topological embedding of $G$. The edge $e'$ is incident with a unique triangle of $G$.  Label each vertex of $e'$ with $0$, and label the third vertex of the
triangle with $1$.
There is now a unique maximal chain $C_{e'}$ in $G$ containing this labeled triangle.
The edges of $C_{e'}$ can be partitioned into two sets: $S_{e'}$ and $L_{e'}$ where
$S_{e'}$ consists of edges of $C_{e'}$ whose vertex labels differ by $1$ and
$L_{e'}$ consists of all edges of $C_{e'}$ whose vertex labels differ by $2$, along with $e'$.

%See Figure~\ref{fi:UChain} for an example of a chain determined by an exterior edge of an  outerplanar graph.
%
%\begin{figure}
%\centering
%\includegraphics[scale=0.75]{figures/GraphWithUChain}
%\caption{\label{fi:UChain}An outerplanar graph with the maximal chain determined by edge $e'$.}
%\end{figure}

Removing the edges of $S_{e'}$ from $G$ produces a set of $2$-connected components in 1-1 correspondence with the edges of  $L_{e'}$: Each component contains exactly one element of  $L_{e'}$ which lies on its outer face.
For each edge $e \in L_{e'}$, let $G_e$ be the component of $G -  S_{e'}$ containing $e$.
We can then recursively decompose each $G_e$ by choosing the (unique) maximal chain $C_e$ in $G_e$ containing the one triangle (if any) of $G_e$ that is incident with $e$. We call the set of chains so constructed a {\em chain decomposition of $G$}. A chain decomposition produces  a decomposition of the edges of $G$ into sets $L$ and $S$, where $L$ is the union of the edges
in each $L_e$ and $S$ is the union of all of the edges in each $S_e$.
Note that there is a single chain for each edge in $L$, and that the collection of chains produced naturally form a tree: The root of the tree is the chain $C_{e'}$ and its children are the chains $\{C_e : e \in L_{e'}\}$; the chain decomposition is entirely determined by the choice of external
edge $e'$.

The drawing algorithm proceeds by first drawing the root chain $C_{e'}$ of the chain decomposition tree of $G$ and then recursively drawing the chain decomposition trees of each $G_e : e \in L_{e'}$. The algorithm depends on a specific method for drawing a single chain. To describe it, we need a few definitions. First, given a line segment $s$ in the plane and a direction
(unit vector) $\mathbf{d}$ not parallel to $s$, denote by $S(s,\mathbf{d})$ the half-infinite strip bounded by $s$ and the two infinite rays in direction $\mathbf{d}$ that have their sources at the endpoints of $s$.  Finally, given a chain $C$, the edges of $C \cap L$ are called {\em external edges} of $C$; note that each external edge is incident to exactly one triangle of $C$.

\begin{lemma}\label{lem:uChainEmbedding}
Given a chain $C$ with $n$ vertices, an external edge $e$ of $C$, a segment $s$ of length $1$ in the plane, and a direction $\mathbf{d}$ such that the (smaller) angle between $s$ and $\mathbf{d}$ is
$\theta < \theta_0 = \arccos(1/4) \approx 75.5^\circ$, there exists a planar straight-line drawing of $C$ such that: (i) the drawing is completely contained within the strip $S(s,\mathbf{d})$; (ii) no external edge $e'$ of $C$ is parallel to $\mathbf{d}$, and the strips $S(e',\mathbf{d})$ are all empty; (iii) each external edge has length $1$ and all other edges have lengths greater than $1/2$; (iv) each external edge forms an angle less than $\theta_0$ with  $\mathbf{d}$. Moreover, such a drawing can be computed in $O(n)$-time in the real RAM model.
\end{lemma}

\begin{proof}
Let $T_0$ be the triangle of $C$ containing $e$. $T_0$ is either adjacent to zero, one, or two triangles of $C$. We handle these three cases in turn. If $T_0$ is the only triangle in $C$, then
we simply draw $T_0$ in $S(s,\mathbf{d})$ as an isosceles triangle with $e$ drawn as $s$ and with its third vertex drawn so that its two edges have length $l$, where $\frac{1}{2} < l < 1$.

Assume now that $T_0$ is adjacent to a triangle $T_1$ of $C$. Denote the vertices of $C$ as follows: $T_0 = \{v_0^-,v_0^+ v_1\}$, where $e = \{v_0^-, v_0^+\}$ and $v_0^-$ is not incident with $T_1$. The vertex of $T_1$ not in $T_0$ is denoted by $v_2$, and, subsequently, the vertex of each $T_i$ not in $T_{i-1}$ is denoted by $v_{i+1}$. We draw $T_0$ as previously, but with more careful positioning of $v_1$. To determine where to position $v_1$, we draw edge $\{v_0^+, v_2\}$ of
$T_1$ as a unit-length segment in direction $\mathbf{d}$. As long as $v_1$ is positioned within
$S(s,\mathbf{d})$ but outside of the disks of radius $\frac{1}{2}$ centered at $v_0^-, v_0^+$, and $v_2$, the edges from each of these vertices to $v_1$ will have length greater than $\frac{1}{2}$
(see top half of Figure 1).
By placing $v_1$ close to $e$, the edges  $\{ v_1, v_0^- \}$ and  $\{ v_1, v_0^+ \}$ will have lengths less than $1$.
Also, since $\angle v_2 v_0^+ v_1 < \theta_0$, edge $\{ v_1, v_2 \}$ will have length less than $1$.

\begin{figure}
\centering
\includegraphics[scale=0.6]{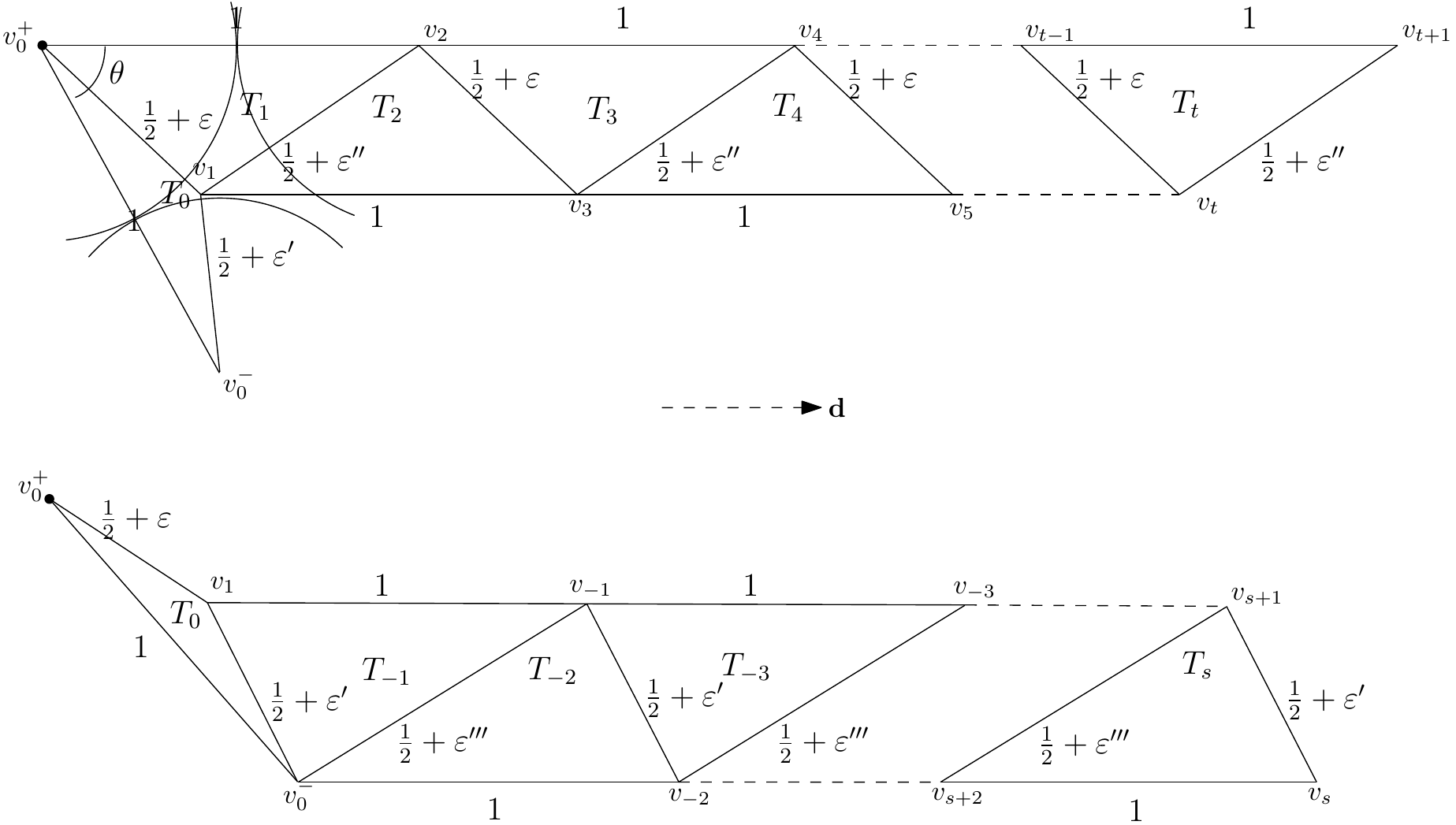}
\caption{\label{fi:EmbeddedUChain}Drawing a chain}
\end{figure}

Assuming that $T_0, \ldots, T_{i-1}$ have been drawn for some $i > 1$, $T_i$ is drawn by positioning $v_{i+1}$  one unit distant  from $v_{i-1}$ in direction $\mathbf{d}$. The result is that
each $T_i$ is congruent to $T_1$ and so the edge-length ratio of $C$ is less than $2$.
At this point, all of the unit-length segments, except for $e$, lie on the two rays in direction
$\mathbf{d}$ emanating from $v_0^+$ and $v_1$. By rotating these rays a very small amount towards one another, we can preserve the lengths of the unit-length segments while ensuring that
all of the remaining segments have lengths in the range $(\frac{1}{2}, 1)$.
See the top of Figure~\ref{fi:EmbeddedUChain}.

Finally, suppose that $T_0$ has two adjacent triangles. Starting with $T_0$, label the other triangles in $C$ so that the labels of adjacent triangles differ by $1$. Thus, for example, $T_0$ is adjacent to $T_1$ and $T_{-1}$. The vertices in the $T_i, i \geq 0$ will be labeled as in the previous case: the unique vertex in $T_i $ not in $T_{i-1}$ is labeled $v_{i+1}$.
The vertices in the $T_i, i < 0$ will be similarly labeled: the unique vertex in $T_i $ not in $T_{i+1}$ is labeled $v_i$.

Draw each $T_i, i > 0$ as in the previous case. Now draw the $T_i, i < 0$ in a similar fashion: Place vertex $v_i$  one unit distant from $v_{i+2}$ in direction $\mathbf{d}$. Then, as above,
all of the unit length edges of the triangles $T_i, i < 0$ will lie on the two rays in direction
$\mathbf{d}$ emanating from $v_1$ and $v_0^-$, and these two rays can be rotated slightly towards each other while maintaining the length of the unit-length edges and ensuring that the other edges still have lengths in the range $(\frac{1}{2}, 1)$.
See the bottom of Figure~\ref{fi:EmbeddedUChain}.
This can clearly be done so that all external edges form angles less than $\theta_0$ with $\mathbf{d}$.

However, we need to ensure that $v_1$ can be placed so that both triangle $T_1$ and triangle $T_{-1}$ can simultaneously satisfy the required edge-length conditions:
Namely, that edges
$\{v_0^-, v_0^+ \}$, $\{v_1, v_{-1} \}$, and $\{v_0^+, v_2 \}$ are all unit-length, while edges
$\{v_1, v_2 \}$, $\{v_0^+, v_1 \}$, $\{v_1, v_0^- \}$, and $\{v_0^-, v_{-1} \}$ all have lengths in the range $(\frac{1}{2}, 1)$.
However, it is relatively simple to show that $v_1$ can always be successfully placed if the (smaller) angle between $s$ and $\mathbf{d}$ is less than
$\theta_0 = \arccos(1/4)$. [The angle $\theta_0$ is the angle opposite an edge of length $2$ in an isosceles triangle having side lengths $2$, $2$, and $1$.]
Finally, the computation of the locations of the vertices can each be computed in constant time in the real RAM model, giving a run-time linear in the size of the chain.\qed
\end{proof}

We are now ready to prove the following lemma. For a planar straight-line drawing $\Gamma$, we denote with $\rho(\Gamma)$ the ratio of a longest to a shortest edge in $\Gamma$.

\begin{lemma}\label{le:upper-bound}
  A maximal outerplanar graph $G$ with $n$ vertices admits a planar straight-line drawing $\Gamma$ with $\rho(\Gamma) < 2$ that can be computed in $O(n)$ time assuming the real RAM model of computation .
\end{lemma}

\begin{proof}
 We call the drawing computed as in Lemma~\ref{lem:uChainEmbedding} a {\em U-strip drawing} of $C$ and adopt the same notation as in Lemma~\ref{lem:uChainEmbedding}.  Recall that in a chain decomposition of a graph, the external edges of the chains are exactly the edges of $L$. A drawing of $G$ is computed as follows.
\begin{enumerate}
\item Compute a chain decomposition tree for $G$; let $C_{e'}$ be the root of the tree.
\item Select a line segment $s$ of length $1$ in the plane and an initial direction $\mathbf{d}$ not parallel to $s$ such that $\angle s \mathbf{d} < \theta_0$.
\item Apply Lemma~\ref{lem:uChainEmbedding} to compute a U-strip drawing of $C_{e'}$.
\item Each edge $e \in L_{e'}$ is drawn as a segment $s_e$ of length $1$, not
parallel to $\mathbf{d}$, that forms an angle with $\mathbf{d}$ that is less than $\theta_0$, so draw the subtree of $C_{e'}$ rooted at $C_e$ in the empty strip $S(s_e, \mathbf{d})$.
\end{enumerate}

The result is an outerplanar straight-line drawing in which all edges of $L$ ({\em long edges}) have length $1$ while all edges in $S$ ({\em short edges}) have length strictly greater than $\frac{1}{2}$.
If we assume that the input is provided to the algorithm in the form of a doubly-connected
edge list~\cite{MULLER1978217}, then a chain decomposition tree for $G$ can be computed in
linear time. Also, since by Lemma~\ref{lem:uChainEmbedding} each chain can be drawn in time proportional to its length, the algorithm
runs in $O(n)$ time in the real RAM model.
\qed
\end{proof}

The following lemma can be proved by means of a packing argument and elementary geometry
(see Appendix A1 for details).

%observing that in any planar straight-line drawing of a maximal outerplanar graph such that  the longest edge has length $1$ and the shortest edge has length at least $\frac{1}{2} + \delta$, the area of every triangular face cannot become arbitrarily small, but it has a lower bound that depends on the value $\delta$. The proof of the lemma applies a packing argument on an outerplanar graph $G$ such that any planar straight-line drawing of $G$ must contain sufficiently many area-disjoint triangles (see the appendix for more details).

\begin{lemma}\label{le:lower-bound}
  For any $\epsilon > 0$ there exists a maximal outerplanar graph whose planar edge length ration is greater than $ 2 - \epsilon$.
\end{lemma}

We conclude the section by observing that Lemmas~\ref{le:upper-bound} and \ref{le:lower-bound} imply Theorem~\ref{th:main}.

\section{Additional Remarks and Open Problems}\label{se:openproblems}

The upper and the lower bound of Theorem~\ref{th:main} suggest some questions that we find worth investigating. One question is whether better bounds on the planar edge-length ratio can be established for subfamilies of outerplanar graphs (for example, it is easy to show that trees have
unit-length drawings). A second question is whether an edge length variance bounded by a constant can be guaranteed for drawings of outerplanar graphs where not all vertices lie in a common face. By a variant of the approach used to prove Lemma~\ref{le:upper-bound} and by using some simple geometric observations, the following results can be proved (see Appendix~A1 and A2 for details).

\begin{theorem}\label{th:bipartite}
  The planar edge-length ratio of a bipartite outerplanar graph is $1$.
\end{theorem}

The \emph{plane edge-length ratio} of a planar embedding $\mathcal{G}$ of a graph $G$ is the minimum edge-length ratio taken over all embedding-preserving planar straight-line drawings of $\mathcal{G}$.

\begin{theorem}\label{th:fixed-embedding}
For any given $k \geq 1$, there exists an embedded outerplanar graph whose plane edge-length ratio is at least $k$.
\end{theorem}

We conclude this paper by listing some open questions that we find interesting to study: (i) Study the edge-length ratio of triangle-free outerplanar graphs. For example, it is not hard to see that if all faces of an outerplanar graph have five vertices, a unit edge length drawing may not exist; however, the planar edge length ratio for this family of graphs could be smaller than the one established in Theorem~\ref{th:main}.
(ii) Extend the result of Theorem~\ref{th:main} to families of non-outerplanar graphs. For example it would be interesting to study whether the planar edge-length ratio of 2-trees is bounded by a constant.
(iii) Study the complexity of deciding whether an outerplanar graph admits a straight-line drawing where the ratio of the longest to the shortest edge is within a given constant. This problem is interesting also in the special case that we want all edges to be unit length.

\newpage
\bibliography{biblio}

\begin{thebibliography}{10}

\bibitem{DBLP:conf/cccg/AichholzerHKR14}
Oswin Aichholzer, Michael Hoffmann, Marc~J. van Kreveld, and G{\"{u}}nter Rote.
\newblock Graph drawings with relative edge length specifications.
\newblock In {\em Proceedings of the 26th Canadian Conference on Computational
  Geometry, {CCCG} 2014, Halifax, Nova Scotia, Canada, 2014}, 2014.

\bibitem{DBLP:journals/ipl/BhattC87}
Sandeep~N. Bhatt and Stavros~S. Cosmadakis.
\newblock The complexity of minimizing wire lengths in {VLSI} layouts.
\newblock {\em Inf. Process. Lett.}, 25(4):263--267, 1987.

\bibitem{DBLP:journals/jgaa/CabelloDR07}
Sergio Cabello, Erik~D. Demaine, and G{\"{u}}nter Rote.
\newblock Planar embeddings of graphs with specified edge lengths.
\newblock {\em J. Graph Algorithms Appl.}, 11(1):259--276, 2007.

\bibitem{Chen2011}
Jianer Chen, Anxiao (Andrew)~Jiang, Iyad~A. Kanj, Ge~Xia, and Fenghui Zhang.
\newblock Separability and topology control of quasi unit disk graphs.
\newblock {\em Wireless Networks}, 17(1):53--67, 2011.

\bibitem{DBLP:books/ph/BattistaETT99}
Giuseppe {Di Battista}, Peter Eades, Roberto Tamassia, and Ioannis~G. Tollis.
\newblock {\em Graph Drawing: Algorithms for the Visualization of Graphs}.
\newblock Prentice-Hall, 1999.

\bibitem{DBLP:journals/siamdm/BattistaV96}
Giuseppe {Di Battista} and Luca Vismara.
\newblock Angles of planar triangular graphs.
\newblock {\em {SIAM} J. Discrete Math.}, 9(3):349--359, 1996.

\bibitem{DBLP:conf/infocom/DohertyPG01}
Lance Doherty, Kristofer S.~J. Pister, and Laurent~El Ghaoui.
\newblock Convex optimization methods for sensor node position estimation.
\newblock In {\em Proceedings {IEEE} {INFOCOM} 2001, The Conference on Computer
  Communications, Twentieth Annual Joint Conference of the {IEEE} Computer and
  Communications Societies, Twenty years into the communications odyssey,
  Anchorage, Alaska, USA, April 22-26, 2001}, pages 1655--1663. {IEEE}, 2001.

\bibitem{DBLP:journals/dam/EadesW90}
Peter Eades and Nicholas~C. Wormald.
\newblock Fixed edge-length graph drawing is {NP}-hard.
\newblock {\em Discrete Applied Mathematics}, 28(2):111--134, 1990.

\bibitem{DBLP:series/natosec/HeldKRV11}
Stephan Held, Bernhard Korte, Dieter Rautenbach, and Jens Vygen.
\newblock Combinatorial optimization in {VLSI} design.
\newblock In Vasek Chv{\'{a}}tal, editor, {\em Combinatorial Optimization -
  Methods and Applications}, volume~31 of {\em {NATO} Science for Peace and
  Security Series - {D:} Information and Communication Security}, pages 33--96.
  {IOS} Press, 2011.

\bibitem{DBLP:conf/cccg/HoffmannKKR14}
Michael Hoffmann, Marc~J. van Kreveld, Vincent Kusters, and G{\"{u}}nter Rote.
\newblock Quality ratios of measures for graph drawing styles.
\newblock In {\em Proceedings of the 26th Canadian Conference on Computational
  Geometry, {CCCG} 2014, Halifax, Nova Scotia, Canada, 2014}, 2014.

\bibitem{DBLP:conf/dialm/KuhnMW04}
Fabian Kuhn, Thomas Moscibroda, and Roger Wattenhofer.
\newblock Unit disk graph approximation.
\newblock In Stefano Basagni and Cynthia~A. Phillips, editors, {\em Proceedings
  of the {DIALM-POMC} Joint Workshop on Foundations of Mobile Computing,
  Philadelphia, PA, USA, October 1, 2004}, pages 17--23. {ACM}, 2004.

\bibitem{MULLER1978217}
D.E. Muller and F.P. Preparata.
\newblock Finding the intersection of two convex polyhedra.
\newblock {\em Theoretical Computer Science}, 7(2):217 -- 236, 1978.

\bibitem{DBLP:conf/usenix/SavareseRL02}
Chris Savarese, Jan~M. Rabaey, and Koen Langendoen.
\newblock Robust positioning algorithms for distributed ad-hoc wireless sensor
  networks.
\newblock In Carla~Schlatter Ellis, editor, {\em Proceedings of the General
  Track: 2002 {USENIX} Annual Technical Conference, June 10-15, 2002, Monterey,
  California, {USA}}, pages 317--327. {USENIX}, 2002.

\end{thebibliography}
\bibliographystyle{plain}

\newpage
\section*{Appendix}

\section*{A1: Proof of Lemma~\ref{le:lower-bound}}

\begin{proof}
  Let the length of the longest edge be $1$. We show that, for any value $\epsilon > 0$, there exists a a maximal outerplanar graph $G$ such that in any planar straight-line drawing of $G$ the length of the shortest edge must be smaller than $\frac{1}{2-\epsilon}$. Let us rewrite  $\frac{1}{2-\epsilon}$ as $\frac{1}{2} + \delta$, where $\delta=\frac{\epsilon}{2(2-\epsilon)}$. In any planar straight-line drawing of a maximal outerplanar graph such that  the longest edge has length $1$ and the shortest edge has length at least $\frac{1}{2} + \delta$, the area of every triangular face cannot become arbitrarily small, but it has a lower bound that depends on the value $\delta$. More precisely, by observing that the minimum area of a triangular face is obtained when one of its sides has length $1$ while the other two have length $\frac{1}{2} + \delta$, by Heron's formula we have that the area of any triangular face under these assumptions is at least $\frac{1}{2} \sqrt{\delta + \delta^2}$.

\begin{figure}
\centering
\includegraphics[scale=0.4]{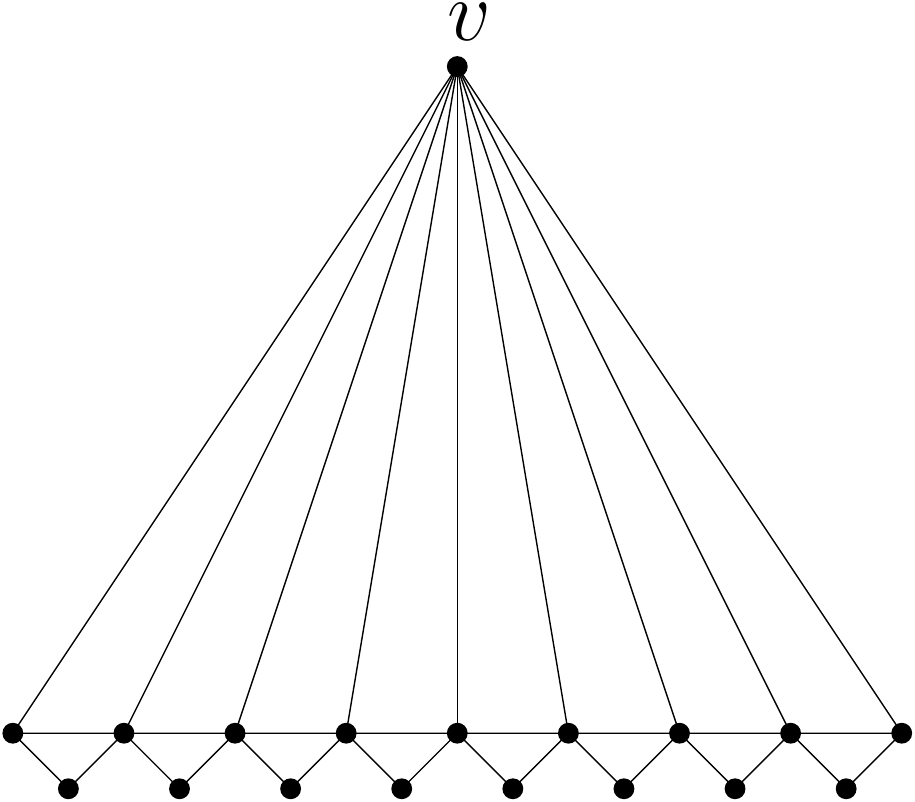}
\caption{\label{fi:lower-bound}An example of the graph in the proof of Lemma \ref{le:lower-bound} when $k =8$.}
\end{figure}

  Let $k$ be an integer such that $k > \frac{8 \pi}{\sqrt{\delta + \delta^2}}$ and consider a fan graph $F$ with a vertex $v$ of degree $k + 2$. $G$ is constructed by adding $k +1$ triangular faces to $F$ as follows: for each edge $e$ of $F$ not incident to $v$, add a new vertex $v_e$ adjacent
to both vertices of $e$. See Figure~\ref{fi:lower-bound} for an illustration when $k = 8$.
We call these non-fan triangles {\em pendant} triangles. Observe that in any planar straight-line drawing $\Gamma$ of $G$ independently of whether all vertices of $\Gamma$ appear on a common face or not, we have that the drawing has at least $k$ area-disjoint pendant triangles, since any pendant triangle that contains another triangle must contain the entire graph. Also, since the longest edge in the drawing has length $1$ and since every vertex has graph theoretic distance at most $2$ from vertex $v$, we have that $\Gamma$ lies inside a disk $D$ of radius $2$  centered at $v$. The number $\nu$ of area disjoint triangles that can be packed inside $D$ such that every triangle has area at least $\frac{1}{2} \sqrt{\delta + \delta^2}$ must have $\nu \leq \frac{8 \pi}{\sqrt{\delta + \delta^2}}$. Since $k > \nu$, it follows that any planar straight-line drawing of the maximal outerplanar graph $G$ where the length of the longest edge is $1$ requires at least one edge having length shorter than $\frac{1}{2} + \delta$. \qed
\end{proof}

\section*{A2: Proof of Theorem 2}
The proof adopts a variant of the approach used to prove Lemma~\ref{le:upper-bound}.
Given a unit-length line segment $s=\overline{p_1 p_2}$ in the plane and two directions  $\mathbf{d_1}$ and  $\mathbf{d_2}$ (unit vectors) in the same half-plane determined by $s$,
consider the region  $W(s,\mathbf{d_1}, \mathbf{d_2})$
bounded by $s$ and the two infinite rays in directions $\mathbf{d_1}$ and  $\mathbf{d_2}$
having their sources at $p_1$ and $p_2$, respectively.
Define $\alpha_W$ to be the sum of the angles $\alpha_1 = \angle \mathbf{d_1} s$ and
$\alpha_2 = \angle s \mathbf{d_2}$, where $\alpha_1$ and $\alpha_2$ are in the half-plane of $s$ containing $\mathbf{d_1}$ and $\mathbf{d_2}$.
We call $W(s,\mathbf{d_1}, \mathbf{d_2})$ a {\em wedge} and $\alpha_W$ the
{\em wedge angle}.
We will show that any bipartite outerplanar graph $G$ admits a unit-length drawing within any wedge
$W(s,\mathbf{d_1}, \mathbf{d_2})$ such that $\alpha_W > \pi$.

To simplify the construction, we first add edges to $G$ until it is a maximal (bipartite)
outerplanar graph and then show that all such graphs admit the desired type of drawing.
Every face of $G$, with the exception of its outerface, is now a quadrilateral.
We call any edge on the non-quadrilateral face an external edge of $G$.
If $G$ itself is a quadrilateral, then every edge of $G$
is external. We are now ready to prove Theorem~\ref{th:bipartite}.

\begin{proof}Let $G$ be a maximal bipartite outerplanar graph, $e = \{u_1,u_4\}$ an external edge of $G$,
$s$ a unit length segment, and $\mathbf{d_1}, \mathbf{d_4}$ two directions
in the same half-plane of $s$ such that $W = W(s,\mathbf{d_1}, \mathbf{d_4})$ has $\alpha_W > \pi$.
We construct an outerplanar straight-line drawing of $G$ within $W$ such that
$e$ is drawn as $s$ and every edge of $G$ has length $1$.

The proof proceeds by induction on the number $f$ of internal faces of $G$.
If $f = 1$, then $G$ can clearly be drawn within $W$ as a rhombus.
So assume now that the result holds for some $f \geq 1$ and that $G$ has $f+1$ internal faces.

Let $e = \{u_1, u_4\}$ and consider the unique quadrilateral face $F = \{ u_1, u_2, u_3, u_4 \}$ containing $e$.
We draw $F$ within $W$ so that $e$ is identified with $s$, and denote by $s_i$ the line segment
drawn for each edge $\{u_i, u_{i+1}\} (1 \leq i \leq 3)$.
Now select directions $\mathbf{d_2}$ and
$\mathbf{d_3}$ for rays emanating from $u_2$ and $u_3$, respectively, such that for $1 \leq i \leq 3$,
each pair $\{\mathbf{d_i}, \mathbf{d_{i+1}}\}$ is in the same half-plane determined by $s_i$, and
$W_i = W(s_i,\mathbf{d_i}, \mathbf{d_{i+1}})$
has an angle greater than $\pi$. Since $\alpha_W > \pi$, such directions exist.
Now separate $G-\{e\}$ into its (at most) three 2-connected components $G_1, G_2, G_3$, where $G_i$ has designated edge $e_i = \{u_i, u_{i+1}\}$.
By induction, each $G_i$ admits an outerplanar straight-line drawing in the corresponding wedge $W_i$ so that $e_i$ is identified with the segment $s_i$. \qed
\end{proof}

\section*{A3: proof of Theorem 3}

In this section we construct a family of embedded outerplanar graphs having unbounded plane edge-length ratios. To prove the correctness of the construction we use the following fact that can be proved by using elementary geometry.

\begin{lemma}\label{le:shrinking-perimeter}
Let $\triangle ABC$ be a triangle with longest edge $\overline{AB}$,
shortest edge $\overline{BC}$, and containing two points $X$ and $Y$
such that the triangles $\triangle ACX$ and $\triangle BCY$ are area-disjoint.
Then at least one of the triangles $\triangle ACX$ and $\triangle BCY$ has perimeter at
least $\frac{| \overline{BC} |}{2}$ shorter than the perimeter of $\triangle ABC$.
\end{lemma}

\begin{proof}
Let $P = a + b + c$ be the perimeter of $ABC$, where $a \leq b \leq c$ are the lengths of the edges opposite the vertices $A, B, C$, respectively.
Consider the segment from $C$ to the (unique) point $D$
on $\overline{AB}$ that divides $\triangle ABC$ into two triangles
$\triangle CBD$ and $\triangle CDA$ of equal perimeter $\frac{P}{2} + d$, where $d$ is the length
of segment $\overline{CD}$.

The difference in perimeter between $\triangle ABC$ and $\triangle CBD$ (or $\triangle CDA$)
is $\frac{P}{2} - d = \frac{a+b+c}{2} - d = \frac{a+b+c-2 d}{2} \geq \frac{a}{2}$
%\[ \frac{P}{2} - d = \frac{a+b+c}{2} - d = \frac{a+b+c-2 d}{2} \geq \frac{a}{2} \]
since $d \leq b \leq c$. \qed
\end{proof}

We are now ready to prove Theorem~\ref{th:fixed-embedding}
\begin{proof}
We construct a family $G_n$ of  outerplanar graphs having planar embeddings $\mathcal{G}_n$ such that $\rho(\mathcal{G}_n) \rightarrow \infty$ as $n \rightarrow \infty$, where
 $\rho(\mathcal{G}_n)$ denotes the plane edge-length ratio of the planar embedding
 $\mathcal{G}_n$ of the graph $G_n$.
Each $G_i$ will have a set of distinguished edges, and for $i > 0$, each $G_i$ will contain $G_{i-1}$ as a subgraph.
$G_0$ is a single triangle, with two distinguished edges. $G_{i+1}$ is constructed from $G_i$ by adding, for each distinguished edge $e$ of $G_i$, a new vertex $v_e$ that is adjacent to both vertices of $e$; the newly added edges incident to each $v_e$ are the distinguished edges of $G_{i+1}$.
There is only $1$ planar embedding of $G_0$; the embedding $\mathcal{G}_1$ of $G_1$
is obtained by putting the two vertices of $G_1 \setminus G_0$ in the inner face of $\mathcal{G}_0$.
For $i > 1$, the embedding $\mathcal{G}_{i+1}$ of $G_{i+1}$ is defined by placing each new vertex $v_e$ in the (unique) triangular face of $\mathcal{G}_i$ having $e$ on its boundary.

Assume that for each $\mathcal{G}_n$ there is a planar straight-line drawing $\Gamma_n$ of
$\mathcal{G}_i$ that preserves the embedding described above, and that, for some
$\rho^* \geq 1, \rho(\Gamma_n) \leq \rho^*$ for all $n \geq 0$.
We show that this assumption leads to a contradiction.
To simplify our notation slightly, we will assume that some edge of $G_0$ has length
$1$, since we can always scale the drawing to make this so;
thus the smallest edge-length in any of the $\Gamma_i$ is at least $\frac{1}{\rho^*}$.

Consider a triangular face $T$ created in the construction of $\mathcal{G}_i$ for some $i>0$.
It consists of a distinguished edge $e$ of $G_{i-1}$ along with the vertex $v_e$ in
$G_i \setminus G_{i-1}$ adjacent to $e$.
The two edges  $e_1$ and $e_2$ incident with $v_e$ are distinguished edges of $G_i$, so in
$\mathcal{G}_{i+1}$ each
of them will form a triangle with some new vertex, say $v_1$ and $v_2$, respectively, and both new vertices will be in the face $T$.

Now consider the line segment $s$ from $v_e$ to $e$
that bisects the perimeter of $T$, as in the proof of Lemma~\ref{le:shrinking-perimeter}, forming
two triangles $T_1$ and $T_2$, where $T_i$ contains $e_i$ on its boundary.
Now either $T_1$ contains $v_1$ or $T_2$ contains $v_2$; whichever $T_j$ contains $v_j$ then
contains the entire triangular face formed by $v_j$ and $e_j$.
Assume, w.l.o.g., that $T_1$ contains $v_1$. We consider two cases.
\begin{enumerate}
\item If $e$ is the longest edge of $T$, then by Lemma~\ref{le:shrinking-perimeter}, the perimeter
of $T_1$ is smaller than that of $T$ by at least half the length of the shortest side of $T$.
But no edge in any of the $\mathcal{G}_i$ is shorter than $\frac{1}{\rho^*}$, and so the perimeter of $T_1$ is
smaller than that of $T$ by at least $\frac{1}{2\rho^*}$.
\item If $e$ is {\em not} the longest edge of $T$, then one of $e_1$ or $e_2$ must be.
If $e' \not= e$ is the longest edge of $T$, then after the construction of $\mathcal{G}_{i+1}$, $e'$ must
be the longest edge of the (new) triangle $T'$ formed by $e'$ and $v_{e'}$ (because $T'$
is contained in $T$, which has $e'$ as its longest edge).
Thus, in the construction of $\mathcal{G}_{i+2}$, one of the new triangles formed, by similar
application of Lemma~\ref{le:shrinking-perimeter}, will have perimeter that is at least
$\frac{1}{2\rho^*}$ shorter than the perimeter of $T'$, which itself has perimeter shorter than
that of $T$.
\end{enumerate}

In each of the two cases above we have identified a triangle, either in $\mathcal{G}_{i+1}$ or in
$\mathcal{G}_{i+2}$,
with perimeter at least $\frac{1}{2\rho^*}$ shorter than that of $T$.
Repeating this process $k$ times results in a triangle that has perimeter at least
$\frac{k}{2\rho^*}$ shorter than that of $T$.
But the perimeter of $T$ is at most $3 \rho^*$, so eventually the quantity
$3 \rho^* - \frac{k}{2\rho^*}$ becomes negative---a contradiction. \qed
\end{proof}

\end{document}